\documentclass{article}%
\usepackage{amssymb}
\usepackage{amsmath}
\usepackage{amsfonts}
\usepackage{graphicx}%
\providecommand{\U}[1]{\protect\rule{.1in}{.1in}}
\newtheorem{theorem}{Theorem}[section]

\newtheorem{conjecture}[theorem]{Conjecture}
\newtheorem{corollary}[theorem]{Corollary}

\newtheorem{lemma}[theorem]{Lemma}

\newtheorem{proposition}[theorem]{Proposition}

\newenvironment{proof}[1][Proof]{\noindent\textbf{#1.} }{\ \rule{0.5em}{0.5em}}
\begin{document}

\title{Critical Sets in Bipartite Graphs}
\author{Vadim E. Levit\\Department of computer Science and Mathematics\\Ariel University Center of Samaria, Israel\\levitv@ariel.ac.il
\and Eugen Mandrescu\\Department of computer Science\\Holon Institute of Technology, Israel\\eugen\_m@hit.ac.il}
\date{}
\maketitle

\begin{abstract}
Let $G=\left(  V,E\right)  $ be a graph. A set $S\subseteq V$ is
\textit{independent} if no two vertices from $S$ are adjacent, and by
$\mathrm{Ind}(G)$ ($\Omega(G)$) we mean the set of all (maximum) independent
sets of $G$, while $\alpha(G)=\left\vert S\right\vert $ for $S\in\Omega(G)$,
and $\mathrm{core}(G)=\cap\{S:S\in\Omega(G)\}$ \cite{LevMan2002a}. The
\textit{neighborhood} of $A\subseteq V$ is\emph{ }denoted by $N(A)=\{v\in
V:N(v)\cap A\neq\emptyset\}$, where $N(v)$ is the neighborhood of the vertex
$v$. The number $d\left(  X\right)  =\left\vert X\right\vert -\left\vert
N(X)\right\vert $ is the \textit{difference} of the set $X\subseteq V$, and
$d_{c}(G)=\max\{d\left(  I\right)  :I\in\mathrm{Ind}(G)\}$ is called the
\textit{critical difference }of $G$. A set $X$ is \textit{critical} if
$d(X)=d_{c}(G)$ \cite{Zhang}.

For a graph $G$ we define $\mathrm{\ker}(G)=\cap\left\{  S:S\text{ \textit{is
a critical independent set}}\right\}  $, while \textrm{diadem}$(G)=\cup
\left\{  S:S\text{ \textit{is a critical independent set}}\right\}  $.

For a bipartite graph $G=\left(  A,B,E\right)  $, with bipartition $\left\{
A,B\right\}  $, Ore \cite{Ore55} defined $\delta\left(  X\right)  =d\left(
X\right)  $ for every $X\subseteq A$, while $\delta_{0}\left(  A\right)
=\max\{\delta\left(  X\right)  :X\subseteq A\}$. Similarly is defined
$\delta_{0}\left(  B\right)  $.

In this paper we prove that for every bipartite graph $G=\left(  A,B,E\right)
$ the following assertions hold:

\begin{itemize}
\item $d_{c}\left(  G\right)  =\delta_{0}\left(  A\right)  +\delta_{0}\left(
B\right)  $;

\item $\mathrm{\ker}(G)=\mathrm{core}(G)$;

\item $\left\vert \mathrm{\ker}(G)\right\vert +\left\vert \mathrm{diadem}%
(G)\right\vert =2\alpha\left(  G\right)  $.

\end{itemize}


\end{abstract}

\section{Introduction}

Throughout this paper $G=(V,E)$ is a simple (i.e., a finite, undirected,
loopless and without multiple edges) graph with vertex set $V=V(G)$ and edge
set $E=E(G)$. If $X\subseteq V$, then $G[X]$ is the subgraph of $G$ spanned by
$X$. By $G-W$ we mean either the subgraph $G[V-W]$, if $W\subseteq V(G)$, or
the partial subgraph $H=(V,E-W)$ of $G$, for $W\subseteq E(G)$. In either
case, we use $G-w$, whenever $W$ $=\{w\}$. By $G=\left(  A,B,E\right)  $ we
denote a bipartite graph having $\left\{  A,B\right\}  $ as a bipartition and
we assume that $A\neq\emptyset\neq B$.

The \textit{neighborhood} of a vertex $v\in V$ is the set $N(v)=\{w:w\in V$
\textit{and} $vw\in E\}$, while the \textit{neighborhood} of $A\subseteq V$
is\emph{ }denoted by $N(A)=N_{G}(A)=\{v\in V:N(v)\cap A\neq\emptyset\}$, and
$N[A]=N(A)\cup A$.

A \textit{matching} is a set of non-incident edges of $G$; a matching of
maximum cardinality $\mu(G)$ is a \textit{maximum matching}, and a
\textit{perfect matching} is a matching covering all the vertices of $G$. If
$M$ is a matching, then $M\left(  v\right)  $ means the mate of the vertex $v$
by $M$, and $M\left(  X\right)  =\left\{  M\left(  v\right)  :v\in X\right\}
$ for $X\subset V\left(  G\right)  $.

A set $S\subseteq V(G)$ is \textit{independent }(or \textit{stable}) if no two
vertices from $S$ are adjacent, and by $\mathrm{Ind}(G)$ we denote the set of
all independent sets of $G$. An independent set of maximum size will be
referred to as a \textit{maximum independent set} of $G$, and the
\textit{independence number }of $G$ is $\alpha(G)=\max\{\left\vert
S\right\vert :S\in\mathrm{Ind}(G)\}$. Let $\Omega(G)$ be the family of all
maximum independent sets of $G$, and $\mathrm{core}(G)=\cap\{S:S\in
\Omega(G)\}$ \cite{LevMan2002a}.

Recall from \cite{Zhang} the following definitions for a graph $G=\left(
V,E\right)  $:

\begin{itemize}
\item $d(X)=\left\vert X\right\vert -\left\vert N(X)\right\vert ,X\subseteq V$
is the \textit{difference} of the set $X$;

\item $d_{c}(G)=\max\{d(X):X\subseteq V\}$ is the \textit{critical difference}
of $G$;

\item a set $U\subseteq V$ is $d$-\textit{critical} if $d(U)=d_{c}(G)$;

\item $id_{c}(G)=\max\{d(I):I\in\mathrm{Ind}(G)\}$ is the \textit{critical
independence difference} of $G$;

\item if $A\subseteq V$ is independent and $d(A)=id_{c}(G)$, then $A$ is
\textit{critical independent}.
\end{itemize}

For a graph $G$ let us denote
\begin{align*}
\mathrm{\ker}(G)  &  =\cap\left\{  S:S\subseteq V\text{ \textit{is a critical
independent set}}\right\}  ,\\
\mathrm{diadem}(G)  &  =\cup\left\{  S:S\subseteq V\text{ \textit{is a
critical independent set}}\right\}  .
\end{align*}

For instance, the graph $G_{1}$ from Figure \ref{fig51} has $X=\left\{
x,y,z,u,v\right\}  $ as a critical set, because $N\left(  X\right)  =\left\{
a,b,u,v\right\}  $ and $d(X)=1=d_{c}(G_{1})$. In addition, let us notice that
$\mathrm{\ker}(G_{1})=\left\{  x,y\right\}  \subset$ \textrm{core}$(G_{1})$,
and \textrm{diadem}$(G_{1})=\left\{  x,y,z\right\}  $. The graph $G_{2}$ from
Figure \ref{fig51} has $d_{c}(G_{1})=d\left(  \left\{  v_{1},v_{2}\right\}
\right)  =\left\vert \left\{  v_{1},v_{2}\right\}  \right\vert -\left\vert
N\left(  \left\{  v_{1},v_{2}\right\}  \right)  \right\vert =1$. It is easy to
see that \textrm{core}$(G_{1})$ is a critical set, while \textrm{core}%
$(G_{2})$ is not a critical set, but $\mathrm{\ker}(G_{2})=\left\{
v_{1},v_{2}\right\}  \subset$ \textrm{core}$(G_{2})$. \begin{figure}[h]
\setlength{\unitlength}{1cm}\begin{picture}(5,1.9)\thicklines
\multiput(1,0.5)(1,0){4}{\circle*{0.29}}
\put(1,0.5){\line(1,0){3}}
\multiput(2,1.5)(1,0){3}{\circle*{0.29}}
\put(2,0.5){\line(0,1){1}}
\put(4,0.5){\line(0,1){1}}
\put(4,0.5){\line(-1,1){1}}
\put(3,1.5){\line(1,0){1}}
\put(3,0.1){\makebox(0,0){$z$}}
\put(1,0.1){\makebox(0,0){$x$}}
\put(4,0.1){\makebox(0,0){$a$}}
\put(2,0.1){\makebox(0,0){$b$}}
\put(1.67,1.5){\makebox(0,0){$y$}}
\put(2.67,1.5){\makebox(0,0){$u$}}
\put(4.3,1.5){\makebox(0,0){$v$}}
\put(0.4,1){\makebox(0,0){$G_{1}$}}
\multiput(6,0.5)(1,0){8}{\circle*{0.29}}
\multiput(6,1.5)(1,0){5}{\circle*{0.29}}
\put(12,1.5){\circle*{0.29}}
\put(6,0.5){\line(1,0){7}}
\put(6,0.5){\line(0,1){1}}
\put(6,0.5){\line(1,1){1}}
\put(7,0.5){\line(1,1){1}}
\put(8,0.5){\line(0,1){1}}
\put(8,0.5){\line(1,1){1}}
\put(9,1.5){\line(1,0){1}}
\put(10,0.5){\line(0,1){1}}
\put(12,0.5){\line(0,1){1}}
\put(12,1.5){\line(1,-1){1}}
\put(5.65,1.5){\makebox(0,0){$v_{1}$}}
\put(6.65,1.5){\makebox(0,0){$v_{2}$}}
\put(6,0.1){\makebox(0,0){$v_{3}$}}
\put(7,0.1){\makebox(0,0){$v_{5}$}}
\put(7.65,1.5){\makebox(0,0){$v_{4}$}}
\put(8,0.1){\makebox(0,0){$v_{6}$}}
\put(9,0.1){\makebox(0,0){$v_{7}$}}
\put(8.65,1.5){\makebox(0,0){$v_{8}$}}
\put(10,0.1){\makebox(0,0){$v_{10}$}}
\put(10.35,1.5){\makebox(0,0){$v_{9}$}}
\put(12.4,1.5){\makebox(0,0){$v_{14}$}}
\put(11,0.1){\makebox(0,0){$v_{11}$}}
\put(12,0.1){\makebox(0,0){$v_{12}$}}
\put(13,0.1){\makebox(0,0){$v_{13}$}}
\put(5.1,1){\makebox(0,0){$G_{2}$}}
\end{picture}\caption{\textrm{core}$(G_{1})=\{x,y,z\}$, while \textrm{core}%
$(G_{2})=\{v_{1},v_{2},v_{7},v_{11}\}$.}%
\label{fig51}%
\end{figure}
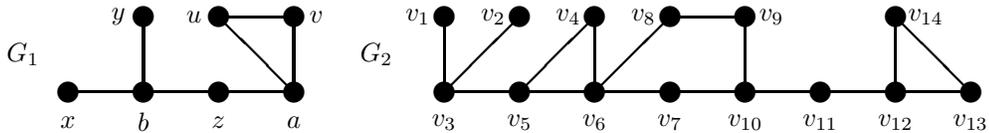

The following results will be used in the sequel.

\begin{theorem}
\label{Th33}Let $G$\ be a graph. Then the following assertions are true:

\emph{(i)} \cite{Zhang} $d_{c}(G)$ $=id_{c}(G)$.

\emph{(ii)} \cite{Larson2007} There is a matching from $N(S)$ into $S$, for
every critical independent set $S$.

\emph{(iii)} \cite{Butenko} Each critical independent set is contained in a
maximum independent set.

\emph{(iv)} \cite{LevMan2011} the function $d$ is supermodular, i.e.,
\[
d(X\cup Y)+d(X\cap Y)\geq d(X)+d(Y)\text{ for every }X,Y\subseteq V\left(
G\right)  ;
\]

\emph{(v)} \cite{LevMan2011} if $S_{1},S_{2}$ are $d$-critical sets, then
$S_{1}\cap S_{2},S_{1}\cup S_{2}$ are $d$-critical as well;

\emph{(vi) }\cite{LevMan2011} there is a unique minimal $d$-critical set,
namely, $\mathrm{\ker}(G)$.

\emph{(vii)} \cite{LevMan2011} $\mathrm{\ker}(G)\subseteq\mathrm{core}(G)$.
\end{theorem}

If $\alpha(G)+\mu(G)=\left\vert V(G)\right\vert $, then $G$ is called a
\textit{K\"{o}nig-Egerv\'{a}ry graph}, \cite{dem}, \cite{ster}. It is
well-known that each bipartite graph enjoys this property \cite{eger},
\cite{koen}.

\begin{theorem}
\label{Th5} If $G=\left(  V,E\right)  $ is a K\"{o}nig-Egerv\'{a}ry graph, $M$
is a maximum matching, and $S\in\Omega\left(  G\right)  $, then:

\emph{(i)} \cite{LevMan2003} $M$ matches $V-S$ into $S$, and $N(\mathrm{core}%
(G))$ into $\mathrm{core}(G)$;

\emph{(ii)} \cite{LevManKE2009} $S$ is $d$-critical, and $d_{c}(G)=\alpha
(G)-\mu(G)=\left\vert \mathrm{core}(G)\right\vert -\left\vert N(\mathrm{core}%
(G))\right\vert $.
\end{theorem}

Following Ore \cite{Ore55}, \cite{Ore62}, the number
\[
\delta(X)=d\left(  X\right)  =\left\vert X\right\vert -\left\vert N\left(
X\right)  \right\vert
\]
is called the \textit{deficiency} of $X$, where $X\subseteq A$ or $X\subseteq
B$ and $G=(A,B,E)$ is a bipartite graph. Let
\[
\delta_{0}(A)=\max\{\delta(X):X\subseteq A\},\quad\delta_{0}(B)=\max
\{\delta(Y):Y\subseteq B\}.
\]

A subset $X\subseteq A$ having $\delta(X)=\delta_{0}(A)$ is called
$A$-\textit{critical}, while $Y\subseteq B$ having $\delta(B)=\delta_{0}(B)$
is called $B$-\textit{critical}. For a bipartite graph $G=\left(
A,B,E\right)  $ let us denote $\mathrm{\ker}_{A}(G)=\cap\left\{  S:S\text{ is
}A\text{-\textit{critical}}\right\}  $ and \textrm{diadem}$_{A}(G)=\cup
\left\{  S:S\text{ is }A\text{-\textit{critical}}\right\}  $. Similarly,
$\mathrm{\ker}_{B}(G)=\cap\left\{  S:S\text{ is }B\text{-\textit{critical}%
}\right\}  $ and \textrm{diadem}$_{B}(G)=\cup\left\{  S:S\text{ is
}B\text{-\textit{critical}}\right\}  $.

It is convenient to define $d\left(  \emptyset\right)  =\delta(\emptyset)=0$.

\begin{figure}[h]
\setlength{\unitlength}{1cm}\begin{picture}(5,1.9)\thicklines
\multiput(5,0.5)(1,0){7}{\circle*{0.29}}
\multiput(4,1.5)(1,0){6}{\circle*{0.29}}
\put(5,0.5){\line(-1,1){1}}
\put(5,0.5){\line(0,1){1}}
\put(5,0.5){\line(1,1){1}}
\put(6,0.5){\line(0,1){1}}
\put(6,0.5){\line(1,1){1}}
\put(7,0.5){\line(-1,1){1}}
\put(7,0.5){\line(0,1){1}}
\put(7,0.5){\line(1,1){1}}
\put(8,0.5){\line(0,1){1}}
\put(8,0.5){\line(1,1){1}}
\put(9,0.5){\line(0,1){1}}
\put(9,1.5){\line(1,-1){1}}
\put(9,1.5){\line(2,-1){2}}
\put(3.65,1.5){\makebox(0,0){$a_{1}$}}
\put(4.65,1.5){\makebox(0,0){$a_{2}$}}
\put(6.35,1.5){\makebox(0,0){$a_{3}$}}
\put(7.35,1.5){\makebox(0,0){$a_{4}$}}
\put(8.35,1.5){\makebox(0,0){$a_{5}$}}
\put(9.35,1.5){\makebox(0,0){$a_{6}$}}
\put(5,0.1){\makebox(0,0){$b_{1}$}}
\put(6,0.1){\makebox(0,0){$b_{2}$}}
\put(7,0.1){\makebox(0,0){$b_{3}$}}
\put(8,0.1){\makebox(0,0){$b_{4}$}}
\put(9,0.1){\makebox(0,0){$b_{5}$}}
\put(10,0.1){\makebox(0,0){$b_{6}$}}
\put(11,0.1){\makebox(0,0){$b_{7}$}}
\put(3,1){\makebox(0,0){$G$}}
\end{picture}\caption{$G$ is a bipartite graph without perfect matchings.}%
\label{fig233}%
\end{figure}
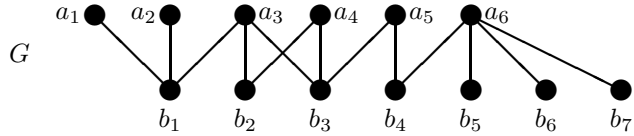For instance, the graph $G=(A,B,E)$ from Figure \ref{fig233} has:
$X=\left\{  a_{1},a_{2},a_{3},a_{4}\right\}  $ as an $A$-critical set,
$\mathrm{\ker}_{A}(G)=\left\{  a_{1},a_{2}\right\}  $, \textrm{diadem}%
$_{A}(G)=\left\{  a_{i}:i=1,...,5\right\}  $ and $\delta_{0}(A)=1$, while
$Y=\left\{  b_{i}:i=4,5,6,7\right\}  $ is a $B$-critical set, $\mathrm{\ker
}_{B}(G)=\left\{  b_{4},b_{5},b_{6}\right\}  $, \textrm{diadem}$_{B}%
(G)=\left\{  b_{i}:i=2,...,7\right\}  $ and $\delta_{0}(B)=2$.

\begin{theorem}
\cite{Ore55}\label{th1} Let $G=\left(  A,B,E\right)  $. Then the following are true:

\emph{(i) }the function $\delta$ is supermodular, i.e., $\delta(X\cup
Y)+\delta(X\cap Y)\geq\delta(X)+\delta(Y)$ for every $X,Y\subseteq A$ (or
$X,Y\subseteq B$).

\emph{(ii)} there is a unique minimal $A$-critical set, namely, $\mathrm{\ker
}_{A}(G)$, and there is a unique maximal $A$-critical set, namely,
\textrm{diadem}$_{A}(G)$; similarly, for $\mathrm{\ker}_{B}(G)$ and
\textrm{diadem}$_{B}(G)$;

\emph{(iii)} $\mu\left(  G\right)  =\left\vert A\right\vert -\delta
_{0}(A)=\left\vert B\right\vert -\delta_{0}(B)$.
\end{theorem}

In this paper we define two new graph parameters, namely, $\ker$ and
\textrm{diadem.} Further, we analyze their relationships with two other
parameters, \textrm{core }and \textrm{corona, }for bipartite graphs.

\section{Preliminaries}

\begin{theorem}
\label{Th4}Let $G=\left(  A,B,E\right)  $ be a bipartite graph. Then the
following assertions are true:

\emph{(i) }$d_{c}(G)=\delta_{0}(A)+\delta_{0}(B)$;

\emph{(ii) }$\alpha\left(  G\right)  =\left\vert A\right\vert +\delta
_{0}(B)=\left\vert B\right\vert +\delta_{0}(A)=\mu\left(  G\right)
+\delta_{0}(A)+\delta_{0}(B)=\mu\left(  G\right)  +d_{c}\left(  G\right)  $;

\emph{(iii)} if $X$ is an $A$-\textit{critical set and }$Y$\textit{ is a }%
$B$-\textit{critical set, then }$X\cup Y$ is a $d$-critical set;

\emph{(iv)} if $Z$ is a $d$-critical independent set, then $Z\cap A$ is an
$A$-\textit{critical set }and $Z\cap B$ is \textit{a }$B$-\textit{critical
set;}

\emph{(v)} if $X$ is either an $A$-\textit{critical set or a }$B$%
-\textit{critical set, then there is a matching from }$N\left(  X\right)  $
into $X$.
\end{theorem}

\begin{proof}
\emph{(i) }By Theorems \ref{th1}\emph{(iii)} and \ref{Th5}\emph{(ii)} we get%
\begin{align*}
d_{c}(G)  &  =\alpha\left(  G\right)  -\mu\left(  G\right)  =\left\vert
A\right\vert +\left\vert B\right\vert -2\mu\left(  G\right)  =\\
&  =\left\vert A\right\vert +\left\vert B\right\vert -\left(  \left\vert
A\right\vert -\delta_{0}(A)\right)  -\left(  \left\vert B\right\vert
-\delta_{0}(B)\right)  =\delta_{0}(A)+\delta_{0}(B).
\end{align*}

\emph{(ii) }Using Theorem \ref{th1}\emph{(iii)}, we infer that\emph{ }%
\[
\alpha\left(  G\right)  =\left\vert A\cup B\right\vert -\mu\left(  G\right)
=\left\vert A\cup B\right\vert -\left\vert A\right\vert +\delta_{0}%
(A)=\left\vert B\right\vert +\delta_{0}(A).
\]
Similarly, one can find $\alpha\left(  G\right)  =\left\vert A\right\vert
+\delta_{0}(B).$

According to part \emph{(i)}, we obtain%
\begin{align*}
\mu\left(  G\right)  +d_{c}\left(  G\right)   &  =\mu\left(  G\right)
+\delta_{0}(A)+\delta_{0}(B)=\\
&  =\left\vert A\right\vert -\delta_{0}(A)+\delta_{0}(A)+\delta_{0}%
(B)=\left\vert A\right\vert +\delta_{0}(B)=\alpha\left(  G\right)  .
\end{align*}

\emph{(iii) }By supermodularity of the function $d$ (Theorem \ref{Th33}%
\emph{(iv)}) and part \emph{(i)}, we have
\begin{align*}
d_{c}\left(  G\right)   &  \geq d(X\cup Y)=d(X\cup Y)+d(X\cap Y)\geq\\
&  \geq d(X)+d(Y)=\delta(X)+\delta(Y)=\delta_{0}(A)+\delta_{0}(B)=d_{c}(G).
\end{align*}

\emph{(iv)} Since $Z=\left(  Z\cap A\right)  \cup\left(  Z\cap B\right)  $ and
$N\left(  Z\cap A\right)  \cap\left(  Z\cap B\right)  =\emptyset=N\left(
Z\cap B\right)  \cap\left(  Z\cap A\right)  $, then%
\begin{gather*}
d\left(  Z\cap A\right)  +d\left(  Z\cap B\right)  =\left\vert Z\cap
A\right\vert -\left\vert N\left(  Z\cap A\right)  \right\vert +\left\vert
Z\cap B\right\vert -\left\vert N\left(  Z\cap B\right)  \right\vert =\\
=\left\vert Z\cap A\right\vert +\left\vert Z\cap B\right\vert -\left\vert
N\left(  Z\cap A\right)  \right\vert -\left\vert N\left(  Z\cap B\right)
\right\vert =\left\vert Z\right\vert -\left\vert N\left(  Z\right)
\right\vert =d\left(  Z\right)  .
\end{gather*}

Using the fact that $d\left(  Z\right)  =d_{c}(G)=\delta_{0}(A)+\delta_{0}%
(B)$, it follows that $d\left(  Z\cap A\right)  =\delta_{0}(A)$ and $d\left(
Z\cap B\right)  =\delta_{0}(B)$.

\emph{(v)} Let $X$ be an $A$-critical set. Suppose to the contrary that there
is no matching from $N\left(  X\right)  $ into $X$. By Hall's Theorem it means
that there exists $U\subseteq N\left(  X\right)  $ such that $\left\vert
N\left(  U\right)  \cap X\right\vert <\left\vert U\right\vert $. Consequently,
we obtain
\begin{gather*}
\delta\left(  X-N\left(  U\right)  \right)  =\left\vert X-N\left(  U\right)
\right\vert -\left\vert N\left(  X-N\left(  U\right)  \right)  \right\vert =\\
=\left\vert X\right\vert -\left\vert X\cap N\left(  U\right)  \right\vert
-\left(  \left\vert N\left(  X\right)  \right\vert -\left\vert U\right\vert
\right)  =\\
=\left\vert X\right\vert -\left\vert N\left(  X\right)  \right\vert +\left(
\left\vert U\right\vert -\left\vert X\cap N\left(  U\right)  \right\vert
\right)  =\delta_{0}(A)+\left(  \left\vert U\right\vert -\left\vert X\cap
N\left(  U\right)  \right\vert \right)  >\delta_{0}(A),
\end{gather*}
which contradicts the fact that $X$ is an $A$-critical set.
\end{proof}

It is known that a bipartite graph$\ G$ has a perfect matching if and only if
$\alpha\left(  G\right)  =\mu\left(  G\right)  $. Hence using Theorem
\ref{Th4}\emph{(ii)}, we deduce the following.

\begin{corollary}
\cite{Ore55} A bipartite graph $G=\left(  A,B,E\right)  $ has a perfect
matching if and only if $\delta_{0}\left(  A\right)  =0=\delta_{0}\left(
B\right)  $.
\end{corollary}

\begin{lemma}
\label{Lem1}Let $G=\left(  A,B,E\right)  $ be a bipartite graph. If $X$ is an
$A$-\textit{critical set and }$Y$\textit{ is a }$B$-\textit{critical set, then
}$\left\vert X\cap N\left(  Y\right)  \right\vert =\left\vert N\left(
X\right)  \cap Y\right\vert $. Moreover, there is a perfect matching between
$X\cap N\left(  Y\right)  $ and $N\left(  X\right)  \cap Y$.
\end{lemma}

\begin{proof}
By Theorem\emph{ }\ref{Th4}\emph{(v)}, there is a matching $M_{1}$ from
$N\left(  X\right)  $ into $X$, and a matching $M_{2}$ from $N\left(
Y\right)  $ into $Y$. For each $b\in N\left(  X\right)  \cap Y$, it follows
that $M_{1}\left(  b\right)  \in N\left(  b\right)  \subseteq N\left(
Y\right)  $ and $M_{1}\left(  b\right)  \in X$. Hence $M_{1}\left(  b\right)
\in X\cap N\left(  Y\right)  $, which implies $M_{1}\left(  N\left(  X\right)
\cap Y\right)  \subseteq X\cap N\left(  Y\right)  $, and further
\[
\left\vert N\left(  X\right)  \cap Y\right\vert =\left\vert M_{1}\left(
N\left(  X\right)  \cap Y\right)  \right\vert \leq\left\vert X\cap N\left(
Y\right)  \right\vert .
\]
Similarly, we have
\[
\left\vert X\cap N\left(  Y\right)  \right\vert =\left\vert M_{2}\left(  X\cap
N\left(  Y\right)  \right)  \right\vert \leq\left\vert N\left(  X\right)  \cap
Y\right\vert .
\]
Consequently, we deduce that $\left\vert X\cap N\left(  Y\right)  \right\vert
=\left\vert N\left(  X\right)  \cap Y\right\vert $ and the restriction of
$M_{1}$ to $N\left(  X\right)  \cap Y$ is a perfect matching from $N\left(
X\right)  \cap Y$ onto $X\cap N\left(  Y\right)  $.
\end{proof}

\begin{corollary}
\label{Cor1} Let $G=\left(  A,B,E\right)  $ is a bipartite graph.

\emph{(i)} \cite{Ore62} If $X=\ker_{A}\left(  G\right)  $ \textit{and }%
$Y$\textit{ is a }$B$-\textit{critical set, then }$X\cap N\left(  Y\right)
=N\left(  X\right)  \cap Y=\emptyset$;

\emph{(ii)} \cite{Ore55} $\ker_{A}\left(  G\right)  \cap N\left(  \ker
_{B}\left(  G\right)  \right)  =N\left(  \ker_{A}\left(  G\right)  \right)
\cap\ker_{B}\left(  G\right)  =\emptyset$.
\end{corollary}

\begin{proof}
\emph{(i) }Assume, to the contrary, that $X\cap N\left(  Y\right)
\neq\emptyset$. By Lemma \ref{Lem1}, we have $\left\vert X\cap N\left(
Y\right)  \right\vert =\left\vert N\left(  X\right)  \cap Y\right\vert $.

If $x\in X-X\cap N\left(  Y\right)  $ has $N\left(  x\right)  \cap
Y\neq\emptyset$, then $x\in N\left(  y\right)  \subseteq N\left(  Y\right)  $,
which is impossible. Hence $N\left(  X-X\cap N\left(  Y\right)  \right)
\subseteq N\left(  X\right)  -N\left(  X\right)  \cap Y$, and further, we get
\begin{gather*}
\left\vert X-X\cap N\left(  Y\right)  \right\vert -\left\vert N\left(  X-X\cap
N\left(  Y\right)  \right)  \right\vert \geq\left\vert X-X\cap N\left(
Y\right)  \right\vert -\left\vert N\left(  X\right)  -N\left(  X\right)  \cap
Y\right\vert =\\
=\left\vert X\right\vert -\left\vert X\cap N\left(  Y\right)  \right\vert
-\left\vert N\left(  X\right)  \right\vert +\left\vert N\left(  X\right)  \cap
Y\right\vert =\delta\left(  X\right)  =\delta_{0}\left(  A\right)  ,
\end{gather*}
and this contradicts the minimality of $X$.

\emph{(ii) }It immediately follows from part \emph{(i)}, when $Y=\ker
_{B}\left(  G\right)  $.
\end{proof}

\section{Ker and Core}

\begin{theorem}
\label{th9}Let $X$ be a critical independent set in a graph $G$. Then the
following statements are equivalent:

\emph{(i) }$X=\mathrm{\ker}(G)$;

\emph{(ii)} there is no set $Y\subseteq$ $N\left(  X\right)  ,Y\neq\emptyset$
such that $\left\vert N\left(  Y\right)  \cap X\right\vert =\left\vert
Y\right\vert $;

\emph{(iii) }for each $v\in X$ there exists a matching from $N\left(
X\right)  $ into $X-v$.
\end{theorem}

\begin{proof}
\emph{(i) }$\Longrightarrow$ \emph{(ii) }By Theorem \ref{Th33}\emph{(ii)},
there is a matching, say $M$, from $N\left(  \mathrm{\ker}(G)\right)  $ into
$\mathrm{\ker}(G)$. Suppose, to the contrary, that there exists some non-empty
set$\ Y\subseteq$ $N\left(  \mathrm{\ker}(G)\right)  $ such that $\left\vert
M\left(  Y\right)  \right\vert =\left\vert N\left(  Y\right)  \cap
\mathrm{\ker}(G)\right\vert =\left\vert Y\right\vert $. It contradicts the
minimality of the set $\mathrm{\ker}(G)$, because
\[
d\left(  \mathrm{\ker}(G)-N\left(  Y\right)  \right)  =d\left(  \mathrm{\ker
}(G)\right)  \text{, while }\mathrm{\ker}(G)-N\left(  Y\right)  \subsetneqq
\mathrm{\ker}(G)\text{.}%
\]

\emph{(ii) }$\Longrightarrow$ \emph{(i) }Suppose $X-\mathrm{\ker}%
(G)\neq\emptyset$. By Theorem \ref{Th33}\emph{(ii)}, there is a matching, say
$M$, from $N\left(  X\right)  $ into $X$. Since there are no edges connecting
vertices from $\mathrm{\ker}(G)$ with vertices of $N\left(  X\right)
-N\left(  \mathrm{\ker}(G)\right)  $, we obtain that $M\left(  N\left(
X\right)  -N\left(  \mathrm{\ker}(G)\right)  \right)  \subseteq X-\mathrm{\ker
}(G)$. Moreover, we have that $\left\vert N\left(  X\right)  -N\left(
\mathrm{\ker}(G)\right)  \right\vert =\left\vert X-\mathrm{\ker}(G)\right\vert
$, otherwise
\begin{align*}
\left\vert X\right\vert -\left\vert N\left(  X\right)  \right\vert  &
=\left(  \left\vert \mathrm{\ker}(G)\right\vert -\left\vert N\left(
\mathrm{\ker}(G)\right)  \right\vert \right)  +\left(  \left\vert
X-\mathrm{\ker}(G)\right\vert -\left\vert N\left(  X\right)  -N\left(
\mathrm{\ker}(G)\right)  \right\vert \right)  >\\
&  >\left(  \left\vert \mathrm{\ker}(G)\right\vert -\left\vert N\left(
\mathrm{\ker}(G)\right)  \right\vert \right)  =d_{c}\left(  G\right)  .
\end{align*}

It means that the set $N\left(  X\right)  -N\left(  \mathrm{\ker}(G)\right)  $
contradicts the hypothesis of \emph{(ii)}, because
\[
\left\vert N\left(  X\right)  -N\left(  \mathrm{\ker}(G)\right)  \right\vert
=\left\vert X-\mathrm{\ker}(G)\right\vert =\left\vert N\left(  N\left(
X\right)  -N\left(  \mathrm{\ker}(G)\right)  \right)  \cap X\right\vert
\text{.}%
\]
Consequently, the assertion is true.

\emph{(ii) }$\Longrightarrow$ \emph{(iii)} By Theorem \ref{Th33}\emph{(ii)},
there is a matching, say $M$, from $N\left(  X\right)  $ into $X$. Suppose, to
the contrary, that there is no matching from $N\left(  X\right)  $ into $X-v$.
By Hall's Theorem, it implies the existence of a set $Y\subseteq$ $N\left(
X\right)  $ such that $\left\vert N\left(  Y\right)  \cap X\right\vert
=\left\vert Y\right\vert $, which contradicts the hypothesis of \emph{(ii)}.

\emph{(iii) }$\Longrightarrow$ \emph{(ii)} Suppose, to the contrary, there is
a non-empty subset $Y$ of $N\left(  X\right)  $ such that $\left\vert N\left(
Y\right)  \cap X\right\vert =\left\vert Y\right\vert $. Let $v$ $\in N\left(
Y\right)  \cap X$. Hence, we get $\left\vert N\left(  Y\right)  \cap
X-v\right\vert <\left\vert Y\right\vert $. Then, by Hall's Theorem, it is
impossible to find a matching from $N\left(  X\right)  $ into $X-v$, which
contradicts the hypothesis of \emph{(iii)}.
\end{proof}

\begin{lemma}
\label{lem1}If $G=\left(  A,B,E\right)  $ is a bipartite graph with a perfect
matching, say $M$, $S\in\Omega\left(  G\right)  $, $X\in$ $\mathrm{Ind}(G)$,
$X\subseteq V\left(  G\right)  -S$, and $G\left[  X\cup M\left(  X\right)
\right]  $ is connected, then
\[
X^{1}=X\cup M\left(  \left(  N\left(  X\right)  \cap S\right)  -M\left(
X\right)  \right)
\]
is an independent set, and $G\left[  X^{1}\cup M\left(  X^{1}\right)  \right]
$ is connected.
\end{lemma}

\begin{proof}
Let us show that the set $M\left(  \left(  N\left(  X\right)  \cap S\right)
-M\left(  X\right)  \right)  $ is independent. Suppose, to the contrary, that
there exist $v_{1},v_{2}\in M\left(  \left(  N\left(  X\right)  \cap S\right)
-M\left(  X\right)  \right)  $ such that $v_{1}v_{2}\in E\left(  G\right)  $.
Hence $M\left(  v_{1}\right)  ,M\left(  v_{2}\right)  \in\left(  N\left(
X\right)  \cap S\right)  -M\left(  X\right)  $.

If $M\left(  v_{1}\right)  $ and $M\left(  v_{2}\right)  $ have a common
neighbor $w\in X$, then $\left\{  v_{1},v_{2},M\left(  v_{2}\right)
,w,M\left(  v_{1}\right)  \right\}  $ spans $C_{5}$, which is forbidden for
bipartite graphs.

Otherwise, let $w_{1},w_{2}\in X$ be neighbors of $M\left(  v_{1}\right)  $
and $M\left(  v_{2}\right)  $, respectively. Since $G\left[  X\cup M\left(
X\right)  \right]  $ is connected, there is a path with even number of edges
connecting $w_{1}$ and $w_{2}$. Together with $\left\{  w_{1},M\left(
v_{1}\right)  ,v_{1},v_{2},M\left(  v_{2}\right)  ,w_{2}\right\}  $ this path
produces a cycle of odd length in contradiction with the hypothesis on $G$
being a bipartite graph.

To complete the proof of independence of the set
\[
X^{1}=X\cup M\left(  \left(  N\left(  X\right)  \cap S\right)  -M\left(
X\right)  \right)
\]
it is enough to demonstrate that there are no edges connecting vertices of $X$
and $M\left(  \left(  N\left(  X\right)  \cap S\right)  -M\left(  X\right)
\right)  $. 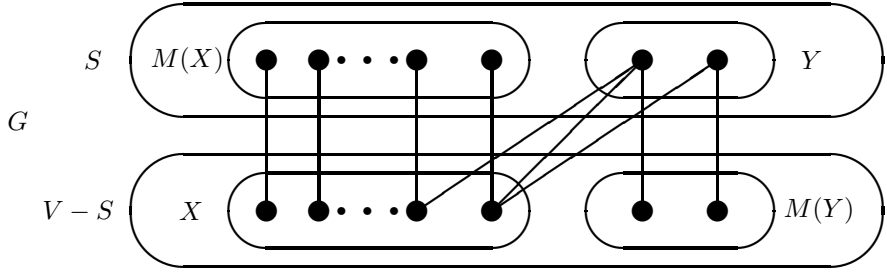
\begin{figure}[h]
\setlength{\unitlength}{1.0cm} \begin{picture}(5,4)\thicklines
\put(6.7,1){\oval(10,1.5)}
\put(5,1){\oval(4,1)}
\put(5,3){\oval(4,1)}
\put(6.7,3){\oval(10,1.5)}
\put(9,1){\oval(2.5,1)}
\put(9,3){\oval(2.5,1)}
\multiput(3.5,1)(0.7,0){2}{\circle*{0.29}}
\multiput(3.5,3)(0.7,0){2}{\circle*{0.29}}
\multiput(5.5,1)(0,2){2}{\circle*{0.29}}
\multiput(6.5,1)(0,2){2}{\circle*{0.29}}
\multiput(4.5,1)(0.35,0){3}{\circle*{0.1}}
\multiput(4.5,3)(0.35,0){3}{\circle*{0.1}}
\put(5.5,1){\line(0,1){2}}
\put(5.5,1){\line(3,2){3}}
\put(6.5,1){\line(0,1){2}}
\put(6.5,1){\line(3,2){3}}
\put(6.5,1){\line(1,1){2}}
\multiput(3.5,1)(0.7,0){2}{\line(0,1){2}}
\put(1,1){\makebox(0,0){$V-S$}}
\put(0.2,2.2){\makebox(0,0){$G$}}
\put(2.5,1){\makebox(0,0){$X$}}
\put(2.45,3){\makebox(0,0){$M(X)$}}
\put(1.2,3){\makebox(0,0){$S$}}
\multiput(8.5,1)(1,0){2}{\circle*{0.29}}
\multiput(8.5,3)(1,0){2}{\circle*{0.29}}
\put(8.5,1){\line(0,1){2}}
\put(9.5,1){\line(0,1){2}}
\put(10.85,1){\makebox(0,0){$M(Y)$}}
\put(10.75,3){\makebox(0,0){$Y$}}
\end{picture}\caption{$S\in\Omega(G)$, $Y=\left(  N\left(  X\right)  \cap
S\right)  -M\left(  X\right)  ${ and }$X^{1}=X\cup M\left(  Y\right)  $.}%
\end{figure}

Assume, to the contrary, that there is $vw\in E$, such that $v\in M\left(
\left(  N\left(  X\right)  \cap S\right)  -M\left(  X\right)  \right)  $ and
$w\in X$. Since $M\left(  v\right)  \in\left(  N\left(  X\right)  \cap
S\right)  -M\left(  X\right)  $ and $G\left[  X\cup M\left(  X\right)
\right]  $ is connected, it follows that there exists a path with an odd
number of edges connecting $M\left(  v\right)  $ to $w$. This path together
with the edges $vw$ and $vM\left(  v\right)  $ produces cycle of odd length,
in contradiction with the bipartiteness of $G$.

Finally, since $G\left[  X\cup M\left(  X\right)  \right]  $ is connected,
$G\left[  X^{1}\cup M\left(  X^{1}\right)  \right]  $ is connected as well, by
definitions of set functions $N$ and $M$.
\end{proof}

Theorem \ref{Th33}\emph{(vii)} claims that $\mathrm{\ker}(G)\subseteq
\mathrm{core}(G)$ for every graph.

\begin{theorem}
\label{th10}If $G$ is a bipartite graph, then $\mathrm{\ker}(G)=\mathrm{core}%
(G)$.
\end{theorem}

\begin{proof}
The assertions are clearly true, whenever $\mathrm{core}(G)=\emptyset$, i.e.,
for $G$ having a perfect matching. Assume that $\mathrm{core}(G)\neq\emptyset$.

Let $S\in\Omega\left(  G\right)  $ and $M$ be a maximum matching. By Theorem
\ref{Th5}\emph{(i)}, $M$ matches $V\left(  G\right)  -S$ into $S$, and
$N(\mathrm{core}(G))$ into $\mathrm{core}(G)$.

According to Theorem \ref{th9}, it is sufficient to show that there is no set
$Z\subseteq N\left(  \mathrm{core}(G)\right)  $, $Z\neq\emptyset$, such that
$\left\vert N\left(  Z\right)  \cap\mathrm{core}(G)\right\vert =\left\vert
Z\right\vert $.

Suppose, to the contrary, that there exists a non-empty set $Z\subseteq$
$N\left(  \mathrm{core}(G)\right)  $ such that $\left\vert N\left(  Z\right)
\cap\mathrm{core}(G)\right\vert =\left\vert Z\right\vert $. Let $Z_{0}$ be a
minimal non-empty subset of $N\left(  \mathrm{core}(G)\right)  $ enjoying this equality.

Clearly, $H=G\left[  Z_{0}\cup M\left(  Z_{0}\right)  \right]  $ is bipartite,
because it is a subgraph of a bipartite graph. Moreover, the restriction of
$M$ on $H$ is a perfect matching.

\textit{Claim 1.} $Z_{0}$ is independent.

Since $H$ is a bipartite graph with a perfect matching it has two maximum
independent sets at least. Hence there exists $W\in\Omega\left(  H\right)  $
different from $M\left(  Z_{0}\right)  $. Thus $W\cap Z_{0}\neq\emptyset$.
Therefore, $N\left(  W\cap Z_{0}\right)  \cap\mathrm{core}(G)=M\left(  W\cap
Z_{0}\right)  $. Consequently,
\[
\left\vert N\left(  W\cap Z_{0}\right)  \cap\mathrm{core}(G)\right\vert
=\left\vert M\left(  W\cap Z_{0}\right)  \right\vert =\left\vert W\cap
Z_{0}\right\vert .
\]
Finally, $W\cap Z_{0}=Z_{0}$, because $Z_{0}$ has been chosen as a minimal
subset of $N\left(  \mathrm{core}(G)\right)  $ such that $\left\vert N\left(
Z_{0}\right)  \cap\mathrm{core}(G)\right\vert =\left\vert Z_{0}\right\vert $.
Since $\left\vert Z_{0}\right\vert =\alpha\left(  H\right)  =\left\vert
W\right\vert $ we conclude with $W=Z_{0}$, which means, in particular, that
$Z_{0}$ is independent.

\textit{Claim 2.} $H$ is a connected graph.

Otherwise, for any connected component of $H$, say $\tilde{H}$, the set
$V\left(  \tilde{H}\right)  \cap Z_{0}$ contradicts the minimality property of
$Z_{0}$.

\textit{Claim 3.} $Z_{0}\cup$ $\left(  \mathrm{core}(G)-M\left(  Z_{0}\right)
\right)  $ is independent.

By Claim 1 $Z_{0}$ is independent. The equality $\left\vert N\left(
Z_{0}\right)  \cap\mathrm{core}(G)\right\vert =\left\vert Z_{0}\right\vert $
implies $N\left(  Z_{0}\right)  \cap\mathrm{core}(G)=M\left(  Z_{0}\right)  $,
which means that there are no edges connecting $Z_{0}$ and $\mathrm{core}%
(G)-M\left(  Z_{0}\right)  $. Consequently, $Z_{0}\cup$ $\left(
\mathrm{core}(G)-M\left(  Z_{0}\right)  \right)  $ is independent.

\textit{Claim 4.} $Z_{0}\cup$ $\left(  \mathrm{core}(G)-M\left(  Z_{0}\right)
\right)  $ is included in a maximum independent set.

Let $Z_{i}=M\left(  \left(  N\left(  Z_{i-1}\right)  \cap S\right)  -M\left(
Z_{i-1}\right)  \right)  ,1\leq i<\infty$. By Lemma \ref{lem1} all the sets
$Z^{i}=\bigcup\limits_{0\leq j\leq i}Z_{j},1\leq i<\infty$ are independent.
Define%
\[
Z^{\infty}=\bigcup\limits_{0\leq i\leq\infty}Z_{i},
\]
which is, actually, the largest set in the sequence $\left\{  Z^{i},1\leq
i<\infty\right\}  $.\begin{figure}[h]
\setlength{\unitlength}{1.0cm} \begin{picture}(5,4)\thicklines
\put(7.5,1){\oval(11,1.5)}
\put(6.7,3){\oval(12.5,1.7)}
\put(3.5,3){\oval(4.35,1.2)}
\put(2.5,3){\oval(1.5,0.8)}
\put(4.5,1){\oval(2,0.8)}
\put(4.5,3){\oval(2,0.8)}
\put(7.5,1){\oval(2,0.8)}
\put(7.5,3){\oval(2,0.8)}
\put(10.5,1){\oval(2,0.8)}
\put(10.5,3){\oval(2,0.8)}
\multiput(4,1)(1,0){2}{\circle*{0.29}}
\multiput(4,3)(1,0){2}{\circle*{0.29}}
\multiput(4,1)(1,0){2}{\line(0,1){2}}
\put(4,1){\line(3,2){3}}
\put(5,1){\line(3,2){3}}
\put(5,1){\line(1,1){2}}
\multiput(7,1)(1,0){2}{\circle*{0.29}}
\multiput(7,3)(1,0){2}{\circle*{0.29}}
\put(7,1){\line(0,1){2}}
\put(8,1){\line(0,1){2}}
\put(7,1){\line(3,2){3}}
\put(8,1){\line(3,2){3}}
\put(8,1){\line(1,1){2}}
\multiput(10,1)(1,0){2}{\circle*{0.29}}
\multiput(10,3)(1,0){2}{\circle*{0.29}}
\put(10,1){\line(0,1){2}}
\put(11,1){\line(0,1){2}}
\multiput(11.9,1)(0.35,0){3}{\circle*{0.1}}
\multiput(11.9,3)(0.35,0){3}{\circle*{0.1}}
\put(0.1,3){\makebox(0,0){$S$}}
\put(1,1){\makebox(0,0){$V-S$}}
\put(0.2,2){\makebox(0,0){$G$}}
\put(0.9,3){\makebox(0,0){$core$}}
\put(3,1){\makebox(0,0){$Z_{0}$}}
\put(4.5,3.1){\makebox(0,0){$Y_{0}$}}
\put(2.5,3.1){\makebox(0,0){$Q$}}
\put(6,1){\makebox(0,0){$Z_{1}$}}
\put(7.5,3.1){\makebox(0,0){$Y_{1}$}}
\put(9,1){\makebox(0,0){$Z_{2}$}}
\put(10.5,3.1){\makebox(0,0){$Y_{2}$}}
\end{picture}\caption{$S\in\Omega(G)$, $Q=\mathrm{core}\left(  G\right)
-M\left(  Z_{0}\right)  $, $Y_{0}=$ $M\left(  Z_{0}\right)  $, $Y_{1}=\left(
N\left(  Z_{0}\right)  -M\left(  Z_{0}\right)  \right)  \cap S$, $Y_{2}=...$
and $Z_{i}=M\left(  Y_{i}\right)  ,i=1,2,...$ {.}}%
\end{figure}
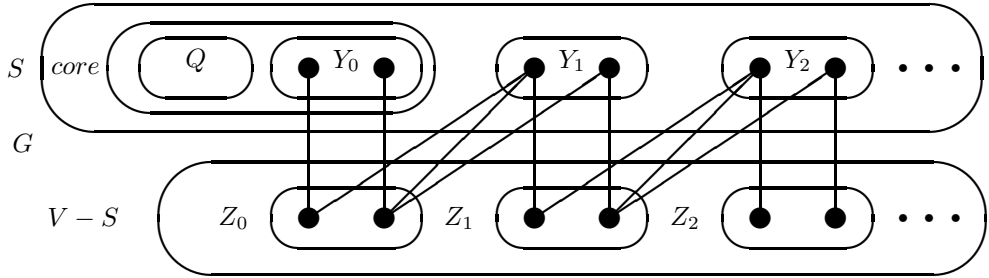

The inclusion
\[
Z_{0}\cup\left(  \mathrm{core}(G)-M\left(  Z_{0}\right)  \right)
\subseteq\left(  S-M\left(  Z^{\infty}\right)  \right)  \cup Z^{\infty}%
\]
is justified by the definition of $Z^{\infty}$.

Since $\left\vert M\left(  Z^{\infty}\right)  \right\vert =\left\vert
Z^{\infty}\right\vert $ we obtain $\left\vert \left(  S-M\left(  Z^{\infty
}\right)  \right)  \cup Z^{\infty}\right\vert =\left\vert S\right\vert $.
According to the definition of $Z^{\infty}$ the set
\[
\left(  N\left(  Z^{\infty}\right)  \cap S\right)  -M\left(  Z^{\infty
}\right)
\]
is empty. In other words, the set $\left(  S-M\left(  Z^{\infty}\right)
\right)  \cup Z^{\infty}$ is independent. Therefore, we arrive at
\[
\left(  S-M\left(  Z^{\infty}\right)  \right)  \cup Z^{\infty}\in\Omega\left(
G\right)  .
\]

Consequently, $\left(  S-M\left(  Z^{\infty}\right)  \right)  \cup Z^{\infty}$
is a desired enlargement of $Z_{0}\cup$ $\left(  \mathrm{core}(G)-M\left(
Z_{0}\right)  \right)  $.

\textit{Claim 5.} $\mathrm{core}(G)\cap\left(  \left(  S-M\left(  Z^{\infty
}\right)  \right)  \cup Z^{\infty}\right)  =\mathrm{core}(G)-M\left(
Z_{0}\right)  $.

The only part of $\left(  S-M\left(  Z^{\infty}\right)  \right)  \cup
Z^{\infty}$ that interacts with $\mathrm{core}(G)$ is the subset
\[
Z_{0}\cup\left(  \mathrm{core}(G)-M\left(  Z_{0}\right)  \right)  .
\]
Hence we obtain
\begin{gather*}
\mathrm{core}(G)\cap\left(  \left(  S-M\left(  Z^{\infty}\right)  \right)
\cup Z^{\infty}\right)  =\\
=\mathrm{core}(G)\cap\left(  Z_{0}\cup\left(  \mathrm{core}(G)-M\left(
Z_{0}\right)  \right)  \right)  =\mathrm{core}(G)-M\left(  Z_{0}\right)  .
\end{gather*}

Since $Z_{0}$ is non-empty, by Claim 5 we arrive at the following
contradiction
\[
\mathrm{core}(G)\nsubseteq\left(  S-M\left(  Z^{\infty}\right)  \right)  \cup
Z^{\infty}\in\Omega\left(  G\right)  .
\]

Finally, we conclude with the fact there is no set $Z\subseteq$ $N\left(
\mathrm{core}(G)\right)  ,Z\neq\emptyset$ such that $\left\vert N\left(
Z\right)  \cap\mathrm{core}(G)\right\vert =\left\vert Z\right\vert $, which,
by Theorem \ref{th9}, means that $\mathrm{core}(G)$ and $\mathrm{\ker}(G)$ coincide.
\end{proof}

Notice that there are non-bipartite graphs enjoying the equality
$\mathrm{\ker}(G)=\mathrm{core}(G)$; e.g., the graphs from Figure \ref{fig14}.
Notice that only $G_{1}$ is a K\"{o}nig--Egerv\'{a}ry graph. 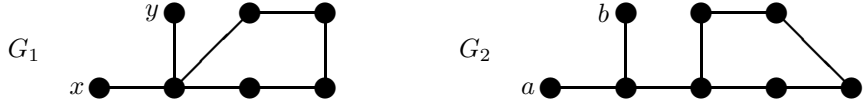
\begin{figure}[h]
\setlength{\unitlength}{1cm}\begin{picture}(5,1.2)\thicklines
\multiput(2,0)(1,0){4}{\circle*{0.29}}
\multiput(3,1)(1,0){3}{\circle*{0.29}}
\put(2,0){\line(1,0){3}}
\put(3,0){\line(0,1){1}}
\put(5,0){\line(0,1){1}}
\put(4,1){\line(1,0){1}}
\put(3,0){\line(1,1){1}}
\put(1.7,0){\makebox(0,0){$x$}}
\put(2.7,1){\makebox(0,0){$y$}}
\put(1,0.5){\makebox(0,0){$G_{1}$}}
\multiput(8,0)(1,0){5}{\circle*{0.29}}
\multiput(9,1)(1,0){3}{\circle*{0.29}}
\put(8,0){\line(1,0){4}}
\put(9,0){\line(0,1){1}}
\put(10,0){\line(0,1){1}}
\put(10,1){\line(1,0){1}}
\put(11,1){\line(1,-1){1}}
\put(7.7,0){\makebox(0,0){$a$}}
\put(8.7,1){\makebox(0,0){$b$}}
\put(7,0.5){\makebox(0,0){$G_{2}$}}
\end{picture}\caption{$\mathrm{core}(G_{1})=\ker\left(  G_{1}\right)
=\{x,y\}$ and $\mathrm{core}(G_{2})=\ker\left(  G_{2}\right)  =\{a,b\}$.}%
\label{fig14}%
\end{figure}

There is a non-bipartite K\"{o}nig-Egerv\'{a}ry graph $G$, such that
$\mathrm{\ker}(G)\neq\mathrm{core}(G)$. For instance, the graph $G_{1}$ from
Figure \ref{fig222} has $\mathrm{\ker}(G_{1})=\left\{  x,y\right\}  $, while
$\mathrm{core}(G_{1})=\left\{  x,y,u,v\right\}  $. The graph $G_{2}$ from
Figure \ref{fig222} has $\mathrm{\ker}(G_{2})=\emptyset$, while $\mathrm{core}%
(G_{2})=\left\{  w\right\}  $.

\begin{figure}[h]
\setlength{\unitlength}{1cm}\begin{picture}(5,1.3)\thicklines
\multiput(4,0)(1,0){4}{\circle*{0.29}}
\multiput(3,1)(1,0){5}{\circle*{0.29}}
\put(4,0){\line(1,0){3}}
\put(4,0){\line(0,1){1}}
\put(3,1){\line(1,-1){1}}
\put(5,0){\line(0,1){1}}
\put(5,0){\line(1,1){1}}
\put(5,1){\line(1,-1){1}}
\put(6,0){\line(0,1){1}}
\put(7,0){\line(0,1){1}}
\put(2.7,1){\makebox(0,0){$x$}}
\put(3.7,1){\makebox(0,0){$y$}}
\put(4.7,1){\makebox(0,0){$u$}}
\put(6.3,1){\makebox(0,0){$v$}}
\put(2,0.5){\makebox(0,0){$G_{1}$}}
\multiput(10,0)(1,0){2}{\circle*{0.29}}
\multiput(12,0)(0,1){2}{\circle*{0.29}}
\put(10,0){\line(1,0){2}}
\put(11,0){\line(1,1){1}}
\put(12,0){\line(0,1){1}}
\put(10,0.3){\makebox(0,0){$w$}}
\put(9,0.5){\makebox(0,0){$G_{2}$}}
\end{picture}\caption{Both $G_{1}$ and $G_{2}$\ are K\"{o}nig-Egerv\'{a}ry
graphs. Only $G_{2}$\ has a perfect matching.}%
\label{fig222}%
\end{figure}
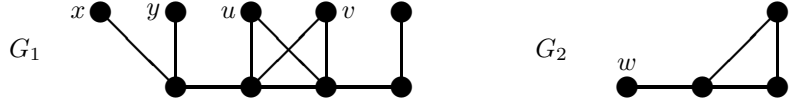

\section{Ker and Diadem}

\begin{proposition}
\cite{LevMan2003}\label{prop10} If $G$ is a K\"{o}nig--Egerv\'{a}ry graph, then

$N\left(  \mathrm{core}(G)\right)  =\cap\{V(G)-S:S\in(G)\}$, i.e., $N\left(
\mathrm{core}(G\right)  )=V\left(  G\right)  -\mathrm{corona}(G)$.
\end{proposition}

There is a non-K\"{o}nig-Egerv\'{a}ry graph $G$ with $V\left(  G\right)
=N\left(  \mathrm{core}(G)\right)  \cup\mathrm{corona}(G)$; e.g., the graph
$G$ from Figure \ref{fig1777}.

\begin{figure}[h]
\setlength{\unitlength}{1cm}\begin{picture}(5,1.2)\thicklines
\multiput(4,0)(1,0){8}{\circle*{0.29}}
\multiput(5,1)(2,0){2}{\circle*{0.29}}
\multiput(8,1)(2,0){2}{\circle*{0.29}}
\put(4,0){\line(1,0){7}}
\put(5,0){\line(0,1){1}}
\put(7,0){\line(0,1){1}}
\put(7,1){\line(1,0){1}}
\put(8,1){\line(1,-1){1}}
\put(10,0){\line(0,1){1}}
\put(10,1){\line(1,-1){1}}
\put(4,0.3){\makebox(0,0){$x$}}
\put(5.3,1){\makebox(0,0){$y$}}
\put(6,0.3){\makebox(0,0){$z$}}
\put(3.2,0.5){\makebox(0,0){$G$}}
\end{picture}
\caption{$G$ is not a K\"{o}nig-Egerv\'{a}ry graph, and $\mathrm{core}%
(G)=\left\{  x,y,z\right\}  $. }%
\label{fig1777}%
\end{figure}
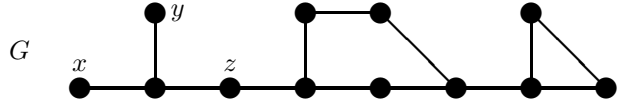

\begin{theorem}
\label{th11}If $G$ is a K\"{o}nig-Egerv\'{a}ry graph, then

\emph{(i)}$\ \left\vert \mathrm{corona}(G)\right\vert +\left\vert
\mathrm{core}(G)\right\vert =2\alpha\left(  G\right)  $;

\emph{(ii) }$\mathrm{diadem}(G)=\mathrm{corona}(G)$, while $\mathrm{diadem}%
(G)\subseteq\mathrm{corona}(G)$ is true for every graph;

\emph{(iii) }$\left\vert \ker\left(  G\right)  \right\vert +\left\vert
\text{\textrm{diadem}}\left(  G\right)  \right\vert \leq2\alpha\left(
G\right)  $.
\end{theorem}

\begin{proof}
\emph{(i) }Using Theorem \ref{Th5}\emph{(ii)} and Proposition \ref{prop10}, we
infer that
\begin{gather*}
\left\vert \mathrm{corona}(G)\right\vert +\left\vert \mathrm{core}%
(G)\right\vert =\left\vert \mathrm{corona}(G)\right\vert +\left\vert N\left(
\mathrm{core}(G)\right)  \right\vert +\left\vert \mathrm{core}(G)\right\vert
-\left\vert N\left(  \mathrm{core}(G)\right)  \right\vert =\\
=\left\vert V\left(  G\right)  \right\vert +d_{c}\left(  G\right)
=\alpha\left(  G\right)  +\mu\left(  G\right)  +d_{c}\left(  G\right)
=2\alpha\left(  G\right)  .
\end{gather*}
as claimed.

\emph{(ii)} Every $S\in\Omega\left(  G\right)  $ is $d$-critical, by Theorem
\ref{Th5}\emph{(ii)}. Further,\emph{ }Theorem \ref{Th5}\emph{(ii)} ensures
that $\mathrm{corona}(G)\subseteq\mathrm{diadem}(G)$. On the other hand, each
critical independent set is included in a maximum independent set, according
to Theorem \ref{Th33}\emph{(iii)}. Thus, we have $\mathrm{diadem}%
(G)\subseteq\mathrm{corona}(G)$. Consequently, the equality $\mathrm{diadem}%
(G)=\mathrm{corona}(G)$ holds.

\emph{(iii)} It follows by combining parts \emph{(i),(ii)} and Theorem
\ref{Th33}\emph{(vii)}.
\end{proof}

Notice that the graph from Figure \ref{fig1777} has $\left\vert
\mathrm{corona}(G)\right\vert +\left\vert \mathrm{core}(G)\right\vert
=13>12=2\alpha\left(  G\right)  $.

For a K\"{o}nig-Egerv\'{a}ry graph with $\left\vert \ker\left(  G\right)
\right\vert +\left\vert \text{\textrm{diadem}}\left(  G\right)  \right\vert
<2\alpha\left(  G\right)  $ see Figure \ref{fig222}. Figure \ref{fig1777}
shows that it is possible for a graph to have $\mathrm{diadem}%
(G)\varsubsetneqq\mathrm{corona}(G)$ and $\mathrm{\ker}(G)\varsubsetneqq
\mathrm{core}(G)$.

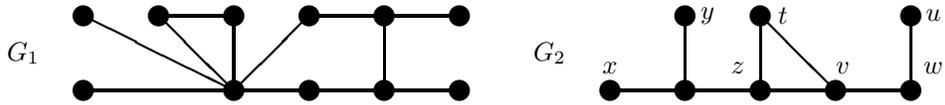
\begin{figure}[h]
\setlength{\unitlength}{1cm}\begin{picture}(5,1.2)\thicklines
\multiput(2,1)(1,0){6}{\circle*{0.29}}
\multiput(4,0)(1,0){4}{\circle*{0.29}}
\put(2,0){\circle*{0.29}}
\put(2,0){\line(1,0){5}}
\put(2,1){\line(2,-1){2}}
\put(3,1){\line(1,-1){1}}
\put(3,1){\line(1,0){1}}
\put(4,0){\line(0,1){1}}
\put(4,0){\line(1,1){1}}
\put(5,1){\line(1,0){2}}
\put(6,0){\line(0,1){1}}
\put(1.2,0.5){\makebox(0,0){$G_{1}$}}
\multiput(9,0)(1,0){5}{\circle*{0.29}}
\multiput(10,1)(1,0){2}{\circle*{0.29}}
\put(13,1){\circle*{0.29}}
\put(9,0){\line(1,0){4}}
\put(10,0){\line(0,1){1}}
\put(11,0){\line(0,1){1}}
\put(11,1){\line(1,-1){1}}
\put(13,0){\line(0,1){1}}
\put(9,0.3){\makebox(0,0){$x$}}
\put(10.3,1){\makebox(0,0){$y$}}
\put(11.3,1){\makebox(0,0){$t$}}
\put(13.3,1){\makebox(0,0){$u$}}
\put(10.7,0.3){\makebox(0,0){$z$}}
\put(12.1,0.3){\makebox(0,0){$v$}}
\put(13.3,0.3){\makebox(0,0){$w$}}
\put(8.2,0.5){\makebox(0,0){$G_{2}$}}
\end{picture}\caption{$G_{1}$ is a non-bipartite K\"{o}nig-Egerv\'{a}ry graph,
such that $\mathrm{\ker}(G_{1})=\mathrm{core}(G_{1})$ and \textrm{diadem}%
$\left(  G_{1}\right)  =\mathrm{corona}(G_{1})$; $G_{2}$ is a
non-K\"{o}nig-Egerv\'{a}ry graph, such that $\mathrm{\ker}(G)=\mathrm{core}%
(G)=\{x,y\}$; \textrm{diadem}$\left(  G\right)  \cup
\{z,t,v,w\}=\mathrm{corona}(G)$.}%
\label{fig17888}%
\end{figure}The combination of $\mathrm{diadem}(G)\varsubsetneqq
\mathrm{corona}(G)$ and $\mathrm{\ker}(G)=\mathrm{core}(G)$ is realized in
Figure \ref{fig17888}.

Now we are ready to describe both $\ker$ and \textrm{diadem} of a bipartite
graph in terms of its bipartition.

\begin{theorem}
Let $G=\left(  A,B,E\right)  $ be a bipartite graph. Then the following
assertions are true:

\emph{(i) }$\ker_{A}\left(  G\right)  \cup$ $\ker_{B}\left(  G\right)
=\ker\left(  G\right)  $;

\emph{(ii) }$\left\vert \ker\left(  G\right)  \right\vert +\left\vert
\text{\textrm{diadem}}\left(  G\right)  \right\vert =2\alpha\left(  G\right)
$;

\emph{(iii) }$\left\vert \ker_{A}\left(  G\right)  \right\vert +\left\vert
\text{\textrm{diadem}}_{B}\left(  G\right)  \right\vert =\left\vert \ker
_{B}\left(  G\right)  \right\vert +\left\vert \text{\textrm{diadem}}%
_{A}\left(  G\right)  \right\vert =\alpha\left(  G\right)  $;

\emph{(iv) }\textrm{diadem}$_{A}\left(  G\right)  \cup$ \textrm{diadem}%
$_{B}\left(  G\right)  =$ \textrm{diadem}$\left(  G\right)  $.
\end{theorem}

\begin{proof}
\emph{(i)} By Theorem \ref{Th4}\emph{(iii)}, $\ker_{A}\left(  G\right)  \cup$
$\ker_{B}\left(  G\right)  $ is $d$-critical in $G$. Moreover, $\ker
_{A}\left(  G\right)  \cup$ $\ker_{B}\left(  G\right)  $ is independent in
accordance with Corollary \ref{Cor1}. Assume that $\ker_{A}\left(  G\right)
\cup$ $\ker_{B}\left(  G\right)  $ is not minimal. Hence the unique minimal
$d$-critical set of $G$, say $Z$, is a proper subset of $\ker_{A}\left(
G\right)  \cup$ $\ker_{B}\left(  G\right)  $, by Theorem \ref{Th33}%
\emph{(vi)}. According to Theorem \ref{Th4}\emph{(iv)}, $Z_{A}=Z\cap A$ is an
$A$-critical set, which implies $\ker_{A}\left(  G\right)  \subseteq Z_{A}$,
and similarly, $\ker_{B}\left(  G\right)  \subseteq Z_{B}$. Consequently, we
get that $\ker_{A}\left(  G\right)  \cup$ $\ker_{B}\left(  G\right)  \subseteq
Z$, in contradiction with the fact that $\ker_{A}\left(  G\right)  \cup$
$\ker_{B}\left(  G\right)  \neq Z\subset\ker_{A}\left(  G\right)  \cup$
$\ker_{B}\left(  G\right)  $.

\emph{(ii), (iii), (iv) }By Corollary \ref{Cor1}, we have
\[
\left\vert \ker_{A}\left(  G\right)  \right\vert -\delta_{0}(A)+\left\vert
\text{\textrm{diadem}}_{B}\left(  G\right)  \right\vert =\left\vert N\left(
\ker_{A}\left(  G\right)  \right)  \right\vert +\left\vert
\text{\textrm{diadem}}_{B}\left(  G\right)  \right\vert \leq\left\vert
B\right\vert .
\]
Hence, according to Theorem \ref{Th4}\emph{(ii)}, it follows that
\[
\left\vert \ker_{A}\left(  G\right)  \right\vert +\left\vert
\text{\textrm{diadem}}_{B}\left(  G\right)  \right\vert \leq\left\vert
B\right\vert +\delta_{0}(A)=\alpha\left(  G\right)  .
\]
Changing the roles of $A$ and $B$, we obtain
\[
\left\vert \ker_{B}\left(  G\right)  \right\vert +\left\vert
\text{\textrm{diadem}}_{A}\left(  G\right)  \right\vert \leq\alpha\left(
G\right)  .
\]

By Theorem \ref{Th4}\emph{(iv)}, \textrm{diadem}$\left(  G\right)  \cap A$ is
$A$-critical and \textrm{diadem}$\left(  G\right)  \cap B$ is $B$-critical.
Hence \textrm{diadem}$\left(  G\right)  \cap A\subseteq$ \textrm{diadem}%
$_{A}\left(  G\right)  $ and \textrm{diadem}$\left(  G\right)  \cap
B\subseteq$ \textrm{diadem}$_{B}\left(  G\right)  $. It implies both the
inclusion $\mathrm{diadem}\left(  G\right)  \subseteq\mathrm{diadem}%
_{A}\left(  G\right)  \cup\mathrm{diadem}_{B}\left(  G\right)  $, and the
inequality
\[
\left\vert \mathrm{diadem}\left(  G\right)  \right\vert \leq\left\vert
\mathrm{diadem}_{A}\left(  G\right)  \right\vert +\left\vert \mathrm{diadem}%
_{B}\left(  G\right)  \right\vert .
\]

Combining Theorem \ref{th10}, Theorem \ref{th11}\emph{(ii)} and part
\emph{(i)} with the above inequalities, we deduce%
\begin{gather*}
2\alpha\left(  G\right)  \geq\left\vert \ker_{A}\left(  G\right)  \right\vert
+\left\vert \ker_{B}\left(  G\right)  \right\vert +\left\vert
\text{\textrm{diadem}}_{A}\left(  G\right)  \right\vert +\left\vert
\text{\textrm{diadem}}_{B}\left(  G\right)  \right\vert \geq\\
=\left\vert \ker\left(  G\right)  \right\vert +\left\vert
\text{\textrm{diadem}}\left(  G\right)  \right\vert =\left\vert
\text{\textrm{core}}\left(  G\right)  \right\vert +\left\vert
\text{\textrm{corona}}\left(  G\right)  \right\vert =2\alpha\left(  G\right)
.
\end{gather*}

Consequently, we infer that%
\begin{gather*}
\left\vert \text{\textrm{diadem}}_{A}\left(  G\right)  \right\vert +\left\vert
\text{\textrm{diadem}}_{B}\left(  G\right)  \right\vert =\left\vert
\text{\textrm{diadem}}\left(  G\right)  \right\vert ,\\
\left\vert \ker\left(  G\right)  \right\vert +\left\vert \text{\textrm{diadem}%
}\left(  G\right)  \right\vert =2\alpha\left(  G\right)  ,\\
\left\vert \ker_{A}\left(  G\right)  \right\vert +\left\vert
\text{\textrm{diadem}}_{B}\left(  G\right)  \right\vert =\left\vert \ker
_{B}\left(  G\right)  \right\vert +\left\vert \text{\textrm{diadem}}%
_{A}\left(  G\right)  \right\vert =\alpha\left(  G\right)  .
\end{gather*}

Since $\mathrm{diadem}\left(  G\right)  \subseteq\mathrm{diadem}_{A}\left(
G\right)  \cup\mathrm{diadem}_{B}\left(  G\right)  $ and $\mathrm{diadem}%
_{A}\left(  G\right)  \cap\mathrm{diadem}_{B}\left(  G\right)  =\emptyset$, we
finally obtain that
\[
\mathrm{diadem}_{A}\left(  G\right)  \cup\mathrm{diadem}_{B}\left(  G\right)
=\mathrm{diadem}\left(  G\right)  ,
\]
as claimed.
\end{proof}

\section{Conclusions}

In this paper we focus on interconnections between $\ker$, \textrm{core,
diadem,} and \textrm{corona} for K\"{o}nig-Egerv\'{a}ry graphs, in general,
and bipartite graphs, in particular.

In \cite{LevManLemma2011} we showed that $2\alpha\left(  G\right)
\leq\left\vert \text{\textrm{core}}\left(  G\right)  \right\vert +\left\vert
\text{\textrm{corona}}\left(  G\right)  \right\vert $ is true for every graph.
By Theorem \ref{th11}\emph{(i)}, this equality is true whenever $G$ is a
K\"{o}nig-Egerv\'{a}ry graph.

According to Theorem \ref{Th33}\emph{(vii)}, $\mathrm{\ker}(G)\subseteq
\mathrm{core}(G)$ for every graph. On the other hand, Theorem \ref{Th33}%
\emph{(iii) }implies the inclusion \textrm{diadem}$\left(  G\right)
\subseteq\mathrm{corona}(G)$. Hence%

\[
\left\vert \ker\left(  G\right)  \right\vert +\left\vert \text{\textrm{diadem}%
}\left(  G\right)  \right\vert \leq\left\vert \text{\textrm{core}}\left(
G\right)  \right\vert +\left\vert \text{\textrm{corona}}\left(  G\right)
\right\vert
\]
for each graph $G$. These remarks together with Theorem \ref{th11}\emph{(iii)}
motivate the following.

\begin{conjecture}
$\left\vert \ker\left(  G\right)  \right\vert +\left\vert
\text{\textrm{diadem}}\left(  G\right)  \right\vert \leq2\alpha\left(
G\right)  $ is true for every graph $G$.
\end{conjecture}

When it is proved one can conclude that the following inequalities:
\[
\left\vert \ker\left(  G\right)  \right\vert +\left\vert \text{\textrm{diadem}%
}\left(  G\right)  \right\vert \leq2\alpha\left(  G\right)  \leq\left\vert
\text{\textrm{core}}\left(  G\right)  \right\vert +\left\vert
\text{\textrm{corona}}\left(  G\right)  \right\vert
\]
hold for every graph $G$.

\end{document}